\documentclass[runningheads]{llncs}
\usepackage{amsmath}
\usepackage{amssymb}
\usepackage{graphicx}
\usepackage[hidelinks]{hyperref}
\usepackage{float}
\usepackage{listings}
\usepackage[scaled]{beramono}
\usepackage[T1]{fontenc}
\usepackage{graphicx}
\newcommand{\Dashv}{\mathrel{\reflectbox{$\vDash$}}}

\newcommand{\rcontr}{\mathsf{C^\rightarrow}}
\newcommand{\lcontr}{\mathsf{C^\leftarrow}}
\newcommand{\rmove}{\mathsf{M^\rightarrow}}
\newcommand{\lmove}{\mathsf{M^\leftarrow}}
\newcommand{\reduces}[2]{\ensuremath{> \hspace{-0.5em} >_{#1,#2}}}

\newcommand{\leftimp}{\ensuremath{\rightarrowtail}}
\newcommand{\rightimp}{\ensuremath{\twoheadrightarrow}}
\newcommand{\nf}{\mathsf{nf}}
\newcommand{\OmegaNF}{\Omega_{\nf}}
\newcommand{\NF}{\mathsf{NF}}

\usepackage{stmaryrd}
\usepackage[square,numbers]{natbib}
\usepackage{lineno}

\usepackage{proof-dashed}
\newcommand{\ms}[1]{\mathsf{#1}}

\newcommand{\mb}[1]{\mathbf{#1}}

\newcommand{\fuse}{\bullet}
\newcommand{\lover}{\mathbin{/}}
\newcommand{\lunder}{\mathbin{\backslash}}
\newcommand{\twist}{\mathbin{\circ}}
\newcommand{\with}{\mathbin{\binampersand}}
\newcommand{\one}{\mb{1}}

\newcommand{\mmode}[1]{{\mathchoice{\ms{#1}}{\ms{#1}}{\scriptscriptstyle\ms{#1}}{\scriptscriptstyle\ms{#1}}}}
\newcommand{\mL}{\mmode{L}}
\newcommand{\mU}{\mmode{U}}

\newcommand{\up}{{\uparrow}}
\newcommand{\down}{{\downarrow}}

\newenvironment{rules}{\[\begin{array}{c}}{\end{array}\]\ignorespacesafterend}
\newenvironment{smallrules}{\begin{small}\[\begin{array}{c}}{\end{array}\]\end{small}\ignorespacesafterend}

\newcommand{\syn}{\Longrightarrow}

\newcommand{\chk}{\Longleftarrow}

\newcommand{\bbar}{\mathord{\Vert}}

\usepackage{xcolor}

\definecolor{keywordcolor}{rgb}{0.0, 0.1, 0.6}    
\definecolor{errorcolor}{rgb}{0.7, 0.1, 0.1}   
\definecolor{matchcolor}{rgb}{0.1, 0.5, 0.1}      
\definecolor{shiftcolor}{rgb}{0.7, 0.4, 0.7} 
\definecolor{requirecolor}{rgb}{0,0.4,0.13} 

\title{Ordered Adjoint Logic}
\author{Sophia Roshal\orcidID{0009-0001-8574-3705} \and
Frank Pfenning \orcidID{0000-0002-8279-5817}}
\authorrunning{S. Roshal and F. Pfenning}
\institute{Carnegie Mellon University}

\usepackage{enumitem}
\setlist{nosep,leftmargin=*}
\begin{document}
\maketitle
\begin{abstract}
  Ordered logics and type systems have been used in a variety of applications
  including computational linguistics, memory allocation, stream processing,
  logical frameworks, parametricity, and enforcing security protocols. In most
  formulations, ordered types are also linear, requiring each resource to be
  used exactly once.  Prior work by Kanovich et al. has investigated calculi
  that relax this constraint through subexponentials within a linear ordered
  logic.  We generalize their work by using adjoint modalities to combine logics
  with varying fine-grained structural properties, including weakening, left
  contraction, right contraction, left mobility, and right mobility.  We show
  that the resulting sequent calculus admits cut elimination.

  We further provide a natural deduction formulation in which structural rules
  are implicit, and show that proof checking for this system is decidable. This
  makes it a suitable foundation for an expressive adjoint programming language
  or logical framework.
\end{abstract}

\section{Introduction}
\label{sec:introduction}

Gentzen's original sequent calculus~\cite{Gentzen35} had explicit structural
rules for exchange, weakening, and contraction.  If we freely allow exchange,
and control the uses of weakening and contraction via an exponential modality,
we obtain either classical or intuitionistic linear logic~\cite{Girard87tcs}.
Prior to that, Lambek~\cite{Lambek58} had considered a sequent calculus
that disallows exchange, driven by applications in linguistics.  A combination of the ideas is natural, and has been
considered for the
classical~\cite{Yetter90jsl,Abrusci99apal,Ruet00mscs,Maieli03iandc} and
intuitionistic~\cite{Polakow99tlca,Polakow99mfps} cases.  The latter has found
applications in logical
frameworks~\cite{Polakow00lfm,Polakow01phd,Pfenning09lics,Simmons12phd,Pfenning23course},
stream processing~\cite{DeYoung16aplas,DeYoung20phd,Cutler24pldi}, and
semantic characterization of programs via
parametricity~\cite{Aberle25fscd}.

In this paper we focus on the intuitionistic variant, where we use the term
\emph{ordered} for logics lacking exchange.  In many applications, we
need to combine several forms of substructural reasoning, which has been
controlled with a modality for mobility $\mbox{\textexclamdown}A$ (effectively
making $A$ linear), and the usual ${!}A$ that also allowing weakening and
contraction~\cite{Polakow99mfps}.  This has been generalized by Kanovich et
al.~\cite{Kanovich18ijcar,Kanovich19mscs} using
\emph{subexponentials}~\cite{Danos93kgc,Nigam16jlc} to give an open-ended
calculus parameterized by a preorder of modes that may or may not satisfy the
structural properties of mobility, weakening, and contraction.

A drawback of subexponentials is that they are interpreted relative to a base
mode in which all the logical inferences take place.  For ordered logics, this
base mode must prohibit all structural rules.  As a result, applications often
require switching into and out of the base mode.  Benton's
LNL~\cite{Benton94csl} offers a solution by allowing linear and nonlinear
reasoning to be connected via two adjoint modalities.  This means we can reason
natively in ordinary or purely linear intuitionistic logic and employ modalities
only where it is necessary to include one in the other.  This has been
generalized to \emph{adjoint
  logic}~\cite{Reed09un,Pruiksma18un,Licata16lfcs,Pruiksma21jlamp,Pruiksma24phd},
essentially by applying Benton's decomposition of the exponential of linear
logic to subexponentials.  A related decomposition has been applied to graded
modal logic~\cite{Hanukaev23tyde,Vollmer25csl}, although not with explicit
control of exchange as far as we are aware.

In this paper we introduce ordered adjoint logic, which features a preorder of
modes, each of which may or may not satisfy mobility, weakening, or contraction.
Instead of a particular base mode (as with subexponentials), we use the
\emph{shift modalities} from adjoint logic to switch between reasoning at
different modes.  We further decompose mobility and contraction to be one-sided,
a possibility anticipated but not explored in the subexponential sequent
calculus of Kanovich et al.~\cite{Kanovich19mscs}.

We study ordered adjoint logic from two perspectives.  We start with a sequent
calculus and show that it satisfies cut and identity elimination, establishing
its bona fides as a well-behaved logical system.  We then introduce a system of
natural deduction for ordered adjoint logic that generalizes the adjoint types
of Jang et al.~\cite{Jang24fscd}, also controlling mobility.  The result of this
latter step is a natural deduction system that can easily be expanded to a type
system for a functional programming language in which one can postulate and
freely combine rather fine-grained information about the structural properties
of code. This extension to programming languages would follow a bidirectional type checking discipline with the set of terms equivalent to those in a non-ordered setting. We show that the proof-checking problem for this system is decidable
even when all applications of structural rules remain implicit.


The same structures recur for a substructural language based on a Kripke
semantics with explicit worlds called QRTT~\cite{Gouni26popl}, which has been applied to
represent and enforce various security properties via typing.  Consequently, our
result applies to ensure decidability of \emph{substructural subsumption}, which
is at the core of the type-checking algorithm for QRTT.

We present an ordered adjoint sequent calculus in \autoref{sec:ordered-seq}.  We
conclude this section with proofs of cut and identity elimination.  In
\autoref{sec:nd}, we introduce two natural deduction calculi: one with
explicit structural rules and one where they remain implicit.  We prove
decidability of proof checking for the latter.
\section{Ordered Adjoint Sequent Calculus}
\label{sec:ordered-seq}

In his seminal work on LNL, Benton~\cite{Benton94csl} connects purely linear
intuitionistic logic with ordinary intuitionistic logic with two adjoint
modalities: $F$ that includes unrestricted formulas in the linear ones, and $G$
that includes the linear formulas in the unrestricted ones.  The beauty of this
approach is that, under the Curry-Howard isomorphism, one can natively program
in either linear or nonlinear discipline, and switch between them only as
needed.  Moreover, we can recover the ${!}A$ as $F\, G\, A$.  This was later
generalized in the form of \emph{adjoint
  logic}~\cite{Reed09un,Pruiksma18un,Pruiksma24phd}, which is parameterized by a
preorder $k \geq m$ of modes, each of which satisfies a set of structural
properties $\sigma(m) \subseteq \{\ms{W}, \ms{C}\}$.  Here, $\ms{W}$
stands for \emph{weakening}, $\ms{C}$ for \emph{contraction}, while
\emph{exchange} is always
assumed.  Structural properties must be monotonic, that is, $k \geq m$ implies
$\sigma(k) \supseteq \sigma(m)$.  Propositions are indexed by modes $A_m$, and
$F$ and $G$ are generalized to \emph{shifts} $(\down^k_m A_k)_m$ and
$(\up_l^m A_l)_m$.  LNL is a special case with two modes $\mU > \mL$,
$\sigma(\mU) = \{\ms{W},\ms{C}\}$ and $\sigma(\mL) = \{\,\}$.  Furthermore,
$G\, A$ becomes $\up_\mL^\mU A_\mL$ and $F\, X$ becomes $\down^\mU_\mL X_\mU$.
We make the presupposition that in a sequent $\Gamma \vdash A_m$, each
antecedent $B_k$ in $\Gamma$ must satisfy $k \geq m$.  In particular,
$B_\mL \vdash A_\mU$, would be semantically unsound: via a cut we could derive
both weakening and contraction for $B_\mL$.  We abbreviate this
\emph{declaration of independence} as $\Gamma \geq m$.  Remarkably, all the
logical rules are parametric in their modes.  Their specific structural
properties only come into play when we have to
decide whether a structural rule can be applied to an antecedent.

In this section we further generalize adjoint logic to account for order among
antecedents.  
Everything explained above about adjoint logic still applies.

\subsection{Structural Rules in an Ordered Setting}
\label{sec:struct-rules}

In Gentzen's presentation of sequent calculus, reordering of antecedents is
governed by the exchange rule, which allows two adjacent elements 
to be swapped. 
\begin{smallrules}
\infer[\mathsf{Exchg}]
{\Omega_L\,A\, B\, \Omega_R \vdash C}
{\Omega_L\,B\, A\, \Omega_R \vdash C}
\end{smallrules}
Exchange acts on two propositions. In the setting of adjoint logic, however,
structural rules apply only when the mode of a proposition admits the
corresponding structural property. This raises the question, when should this swap be allowed?

An elegant solution is offered by Kanovich et
al.~\cite{Kanovich19mscs}: rather than trying to enforce
some reasonable condition on this rule they use \emph{mobility}. Instead of
swapping two propositions, mobility rules move a single proposition arbitrarily
far to the left or right within the antecedents. This formulation aligns
naturally with the unary presentation of weakening and contraction. Each
mobility rule is permitted precisely when the mode of the proposition admits the
corresponding directional mobility.  New here is that we treat left and right
mobility separately.

Since it is important to track specific occurrences of hypotheses in sequents, we
label each antecedent with a variable.  When new antecedents are introduced in
the premises of rules, we choose names not already present in the sequent.
Here, and subsequently, we use $\Omega$ to stand for the collection of
antecedents instead of $\Gamma$ in order to emphasize that their order is
significant.
\begin{smallrules}
  \infer[\mathsf{M}^\leftarrow]
  {\Omega_L\, (x{:}A_m)\, \Omega_M\, \Omega_R \vdash C_r}
  {\Omega_L\, \Omega_M\, (x{:}A_m)\, \Omega_R \vdash C_r
    & \mathsf{M}^\leftarrow \in \sigma(m)}
  \\[0.75em]
  \infer[\mathsf{M^\rightarrow}]
  {\Omega_L\, \Omega_M\, (x{:}A_m)\, \Omega_R \vdash C_r}
  {\Omega_L\, (x{:}A_m)\, \Omega_M\, \Omega_R \vdash C_r
    & \mathsf{M}^\rightarrow \in \sigma(m)}
\end{smallrules} 
If all antecedents have the same structural properties, left mobility implies
right mobility and vice versa.  However, immobile antecedents can block
exchange, so the two structural properties become independent.  This has
applications, for example, in modeling protocols and security properties (see
\autoref{sec:conclusion}).

As observed by Kanovich et al., if we restrict contraction to two adjacent
antecedents the resulting calculus does not satisfy cut elimination.
We therefore admit non-local contraction in both
directions.  As with mobility, we maintain two separate structural properties.
\begin{smallrules}
  \infer[\mathsf{C}^\leftarrow]
  {\Omega_L\, (x{:}A_m)\, \Omega_M\, \Omega_R \vdash C_r}
  {\Omega_L\, (x{:}A_m)\, \Omega_M\, (x{:}A_m)\, \Omega_R \vdash C_r
    & \mathsf{C}^\leftarrow \in \sigma(m)}
  \\[0.75em]
  \infer[\mathsf{C}^\rightarrow]
  {\Omega_L\, \Omega_M\, (x{:}A_m)\, \Omega_R \vdash C_r}
  {\Omega_L\, (x{:}A_m)\, \Omega_M\, (x{:}A_m)\, \Omega_R \vdash C_r
    & \mathsf{C}^\rightarrow \in \sigma(m)}
\end{smallrules}
Tracking contraction in this manner with duplicate variables will be important
in the proof of cut admissibility and in \autoref{sec:nd}.  While duplicate
variables are allowed, all occurrences of a variable must label the same
proposition, and variables introduced in the premises of rules of must be fresh.
Weakening remains a single rule:
\begin{smallrules}
  \infer[\mathsf{W}]
  {\Omega_L\, (x{:}A_m)\, \Omega_R \vdash C_r}
  {\Omega_L\Omega_R \vdash C_r & \mathsf{W} \in \sigma(m)}
\end{smallrules}

\subsection{Cut and Identity}

The rule of identity is unremarkable.  The rule of cut requires a condition on
modes so that our presupposition is preserved when reading the rule from the
conclusion to the premises.  We will shortly have occasion to generalize cut in
order to prove its admissibility in the cut-free sequent calculus.
\begin{smallrules}
  \infer[\ms{id}]
  {x{:} A_m \vdash A_m}
  {}
  \hspace{3em}
  \infer[\ms{cut}]
  {\Omega_L\, \Omega\, \Omega_R \vdash C_r}
  {\Omega \vdash A_m
    & \Omega_L\, (x{:} A_m)\, \Omega_R \vdash C_r
    & (\Omega \geq m \geq r)}
\end{smallrules}

\subsection{Logical Rules}

As compared to linear logic, implication splits into two connectives: left
implication $A_m \leftimp B_m$ (written $A \lunder B$ in the Lambek
calculus~\cite{Lambek58}) and right implication $A_m \rightimp B_m$ (usually
written $B \lover A$). It is also possible to split conjunction into fuse
($A_m \fuse B_m$) and twist ($A_m \twist B_m$).  The latter is usually omitted
since $A_m \fuse B_m$ is isomorphic to $B_m \twist A_m$.
We therefore obtain the following language of propositions, where $P_m$ stands
for atomic propositions. Note that in the first row, all non-atomic propositions have rules which are invertible on the right and therefore are negative, while in the second row they have rules which are invertible on the left and therefore are positive. 
\[
  \begin{array}{llcll}
    & A_m, B_m & ::=
    & P_m \mid A_m \leftimp B_m \mid  A_m \rightimp B_m \mid \with \{\ell : A_m^\ell\}_{\forall \ell \in L}
      \mid \up^m_l A_l & (m \geq l) \\ 
                 & & \mid
    & A_m \fuse B_m \mid \one_m \mid \oplus\{A_m^\ell\}_{\forall \ell \in L} 
      \mid \down^k_m A_k & (k \geq m) 
  \end{array}
\]
Note that $\oplus$ and $\with$ are generalized to be $n$-ary rather than binary. 
This means they could also be nullery or empty. In a sequent calculus, non-algorithmic presentation, 
this does not cause any difficulties. Later in the natural deduction section, we will discuss the difficulties of these empty types when developing an algorithm.
Remarkably, as for unordered adjoint logic, the left and right rules for the
connectives are entirely parametric in their mode.  We just have to be more
careful about maintaining order.  We highlight the rules for left and right
implication.  The complete set of rules can be found in \autoref{fig:ordered-seq}.
\begin{smallrules}
  \infer[\rightimp R]
  {\Omega \vdash A_m \rightimp B_m}
  {\Omega\, (x{:}A_m) \vdash B_m }
  \hspace{3em}
  \infer[\rightimp L]
  {\Omega_L\, (f{:}A_m \rightimp B_m)\, \Omega_A\, \Omega_R \vdash C_r}
  {\Omega_A \vdash A_m
    & \Omega_A \geq m & \Omega_L\, (y{:}B_m)\, \Omega_R \vdash C_r}
  \\[0.75em]
  \infer[\leftimp R]
  {\Omega \vdash A_m \leftimp B_m}
  {(x{:}A_m)\, \Omega \vdash B_m }
  \hspace{3em}
  \infer[\leftimp L]
  {\Omega_L\, \Omega_A\, (f{:}A_m \leftimp B_m)\, \Omega_R \vdash C_r}
  {\Omega_A \vdash A_m & \Omega_A \geq m 
    & \Omega_L\, (y{:}B_m)\, \Omega_R \vdash C_r }
\end{smallrules}%
\begin{figure}
    \centering
    \begin{smallrules}
    \infer[1R]{\cdot \vdash 1_m}{}
    \hspace{3em}
    \infer[1L]{\Omega_L\, (x{:}1_m)\, \Omega_R \vdash z{:} C_r}{\Omega_L\, \Omega_R \vdash C_r}
    \\[0.75em]
    \infer[\fuse R]{\Omega_1\, \Omega_2 \vdash A_m \fuse B_m}{\Omega_1 \vdash A_m & \Omega_2 \vdash B_m}
    \hspace{3em}
    \infer[\fuse L]
    {\Omega_L\, (x{:} A_m \fuse B_m)\, \Omega_R \vdash C_r}
    {\Omega_L\, (x_1{:}A_m)\, (x_2{:} B_m)\, \Omega_R \vdash C_r} 
    \\[0.75em]
    \infer[\oplus R]{\Omega \vdash \oplus \{\ell : A_m^\ell\}_{\forall \ell \in L}}{k \in L & \Omega \vdash A_m^k}\hspace{3em}
    \infer[\oplus L]
    {\Omega_L\, (x{:} \oplus \{\ell : A_m^\ell\}_{\forall \ell \in L})\, \Omega_R\vdash C_r}
    {\Omega_L\, (y{:} A_m^\ell)\, \Omega_R \vdash C_r
      \forall \ell \in L}
    \\[0.75em]

    \infer[\down R]
    {\Omega \vdash \down^k_m A_k}
    {\Omega \vdash A_k & (\Omega \geq k) }
    \hspace{3em}
    \infer[\down L]
    {\Omega_L\, (x{:} \down^k_m A_k)\, \Omega_R \vdash C_r}
    {\Omega_L\, (y{:} A_k)\, \Omega_R \vdash C_r}
    \\[0.75em]\hline\\[0.75em]
    \infer[\rightimp R]{\Omega \vdash A_m \rightimp B_m}{\Omega\, (x{:}A_m) \vdash B_m }
    \hspace{3em}
    \infer[\rightimp L]
    {\Omega_L\, (f{:}A_m \rightimp B_m)\, \Omega_A\, \Omega_R \vdash C_r}
    {\Omega_A \vdash A_m & \Omega_A \geq m & \Omega_L\, (y{:}B_m)\, \Omega_R  \vdash C_r}
    \\[0.75em]
    \infer[\leftimp R]{\Omega \vdash A_m \leftimp B_m}{(x{:}A_m)\, \Omega \vdash B_m}
    \hspace{3em}
    \infer[\leftimp L]
    {\Omega_L\, \Omega_A\, (f{:}A_m \leftimp B_m)\, \Omega_R \vdash C_r}
    {\Omega_A \vdash A_m & \Omega_A \geq m & \Omega_L\, (y{:}B_m)\, \Omega_R \vdash C_r }
    \\[0.75em]
    \infer[\with R]{\Omega \vdash A_m \with B_m}{\Omega \vdash A_m & \Omega \vdash B_m}\\[0.75em]
    \infer[\with L_1]
    {\Omega_L\, (x{:}A_m\with B_m)\, \Omega_R \vdash C_r}
    {\Omega_L\, (y{:} A_m)\, \Omega_R \vdash C_r}
    \hspace{3em}
    \infer[\with L_2]
    {\Omega_L\, (x{:}A_m\with B_m)\, \Omega_R \vdash C_r}
    {\Omega_L\, (y{:} B_m)\, \Omega_R \vdash C_r}
    \\[0.75em]
    \infer[\up R]{\Omega \vdash \up_l^m A_l}{\Omega \vdash A_l}
    \hspace{3em}
    \infer[\up L]
    {\Omega_L\, (x{:}\up_l^m A_l)\, \Omega_R \vdash C_r}
    {\Omega_L\, (y{:}A_l)\, \Omega_R \vdash C_r & (l \geq r) }
    \\[0.75em]\hline\\[0.75em]  
    
    \infer[\mathsf{id}]{x{:}A_m \vdash A_m}{}
    \hspace{3em}
    \infer[\mathsf{cut}]
    {\Omega_L\, \Omega_A\, \Omega_R \vdash C_r}
    {\Omega_A \vdash A_m & \Omega_L\, (x{:} A_m)\, \Omega_R \vdash C_r & (\Omega \geq m \geq r)}
  \end{smallrules}
    \caption{Ordered Adjoint Sequent Calculus (structural rules in \autoref{sec:struct-rules})}
    \label{fig:ordered-seq}
\end{figure}%
The rules regarding shifts are trickier, but the conditions on modes are
systematically derived from the need to preserve \emph{independence} when
reading the rules from the conclusion to the premises.

The rules for the downshift are:
\begin{smallrules}
  \infer[\down L]
  {\Omega_L\, (x{:}\down^k_m A_k)\, \Omega_R \vdash C_r}
  {\Omega_L\, (y{:}A_k)\, \Omega_R \vdash C_r}
  \hspace{3em}
  \infer[\down R]
  {\Omega \vdash \down^k_m A_k}
  {\Omega \vdash A_k & (\Omega \geq k)}
\end{smallrules}
In the left rule, the presuppositions $\Omega_L\,\Omega_R \geq r$ and
$k \geq m \geq r$ ensure that the premise is well formed. In the right rule,
$\Omega \geq m$ and $k \geq m$ do not imply $\Omega \geq k$, so this condition
must be checked explicitly.

The upshift rules are:
\begin{smallrules}
  \infer[\up L]
  {\Omega_L\, (x{:} \up^m_l A_l)\, \Omega_R \vdash C_r}
  {\Omega_L\, (y{:} A_l)\, \Omega_R \vdash C_r & (l \geq r)}
  \hspace{3em}
  \infer[\up R]
  {\Omega \vdash \up^m_l A_l}
  {\Omega \vdash A_l}
\end{smallrules}
For the left rule, the condition $l \geq r$ must be verified explicitly, since
it does not follow from $\Omega_L\, \Omega_R \geq r$, $m \geq l$ and $m \geq
r$. The right rule has no additional conditions on modes.
To show where unidirectional mobility makes
a difference, consider the following proof{:}

\begin{smallrules}
     \infer[\rightimp R]
  {\cdot \vdash \down^k_m A_k \rightimp B_m \rightimp B_m \fuse \down^k_m A_k}
  {\infer[\rightimp R]
    {x{:}\down^k_m A_k \vdash B_m \rightimp B_m \fuse \down^k_m A_k}
    {\infer[\down L]
      {(x{:}\down^k_m A_k)\, (y{:}B_m) \vdash B_m \fuse \down^k_m A_k}
      {\infer[\lmove]
        {(x' {:} A_k)\, (y{:}B_m) \vdash B_m \fuse \down^k_m A_k}
        {\infer[\fuse R]
          {(y{:}B_m)\, (x'{:}A_k) \vdash B_m \fuse \down^k_m A_k}
          {\infer[\ms{id}]{y{:}B_m \vdash B_m}{}
            & \infer[\down R]
            {x'{:}A_k \vdash \down^k_m A_k}
            {\infer[\ms{id}]{x'{:}A_k \vdash A_k}{}
              & (k \geq k)}}
          & \lmove \in \sigma(k)}}}}
\end{smallrules}
 
\noindent This proof  goes through if $k$ admits left mobility ($m$ need not admit any
mobility), or if $m$ admits any mobility (in which case we replace the $\lmove$
of $x'$ with an $\rmove$ of $y$). However, there is no proof if $m$ is immobile
and $k$ only admits right mobility.

\subsection{Cut Elimination}

Because we have explicit contraction rules, with multiple copies of the same
variable appearing in a context, we use a form of multicut rather than a
single cut \cite{Gentzen35,Girard87tcs,Pruiksma24phd}.  The number $n$ of
occurrences of a variable replaced in the antecedents must be compatible with
the structural properties of the mode $m$, written $|m| \sim n$.
Specifically{:}
\[
  \begin{array}{lcll}
    |m| & \sim & 0 & \mbox{if $\ms{W} \in \sigma(m)$} \\
    |m| & \sim & 1 & \mbox{always} \\
    |m| & \sim & n & \mbox{for $n > 1$ if $\lcontr \in \sigma(m)$ or $\rcontr \in \sigma(m)$}
  \end{array}
\]
In multicut, $\Omega(x{:}A_m)^n$ denotes $n$ occurrences of $x{:}A_m$
distributed throughout $\Omega$, and $\Omega(\Omega_A)^n$ denotes their
replacement by $\Omega_A$.  By our general assumptions, the variable $x$ must be
fresh in the premise.
\begin{smallrules}
  \infer[\mathsf{mcut}]
  {\Omega(\Omega_A)^n \vdash C_r}
  {\Omega_A \vdash A_m &
    \Omega\, (x{:}A_m)^n \vdash C_r
    & (\Omega_A \geq m \geq r)
    & |m| \sim n}
\end{smallrules}

\noindent We can now state and prove the admissibility of multicut in the
cut-free calculus, followed by general cut elimination.

\begin{theorem}[Admissibility of Multicut]
\label{thm:mcut-admit}
If $\Omega_A \vdash A_m$ and $\Omega(x{:}A_m)^n \vdash C_r$ with cut-free
derivations, then $\Omega(\Omega_A)^n \vdash C_r$ has a cut-free derivation.
\end{theorem}
\begin{proof}
  We proceed by induction on the ordered triple $(A_m,\mathcal{D}, \mathcal{E})$
  We show the principal cases for weakening, and one direction of contraction
  and mobility.  We use dashed lines to indicate admissibility.
\begin{description}
\item[Case:] $\mathcal{E}$ ends in $\mathsf{W}$ on one occurrence of the principal formula.
\begin{smallrules}
  \infer-[\mathsf{mcut}]
  {(\Omega_L\, \Omega_A\, \Omega_R)(\Omega_A)^{n-1} \vdash C_r}
  {\deduce{\Omega_A \vdash A_m}{\mathcal{D}}
    & \infer[\mathsf{W}]
    {(\Omega_L\, (x{:}A_m)\, \Omega_R)(x{:}A_m)^{n-1} \vdash C_r}
    {\deduce{(\Omega_L\, \Omega_R)(x{:}A_m)^{n-1} \vdash C_r}{\mathcal{E}'}
      & \deduce{\mathsf{W} \in \sigma(m)}{(2)}}
    & \deduce{\Omega_A \geq m \geq r}{(1)}}
\end{smallrules}
\begin{tabbing}
    $(\Omega_L\, \Omega_R)(\Omega_A)^{n-1} \vdash C_r$ \` by IH on $(A_m,\mathcal{D},\mathcal{E}')$, with (1)\\
    $\Omega_A$ admits weakening \` by monotonicity ((2),(1))\\
    $(\Omega_L\, \Omega_A\, \Omega_R)(\Omega_A)^{n-1} \vdash C_r$ \` by $\mathsf{W}$ on elements of $\Omega_A$
\end{tabbing}

\item[Case:] $\mathcal{E}$ ends in $\lcontr$ on (one copy of) principal formula.
\begin{smallrules}
  \infer-[\mathsf{mcut}]
  {(\Omega_L\, \Omega_A\, \Omega_M\, \Omega_R)(\Omega_A)^{n-1} \vdash C_r}
  {\deduce{\Omega_A \vdash A_m}{\mathcal{D}}
    & \infer[\lcontr]
    {(\Omega_L\, (x{:}A_m)\, \Omega_M\, \Omega_R)(x{:}A_m)^{n-1} \vdash C_r}
    {\deduce{(\Omega_L\, (x{:}A_m)\, \Omega_M\, (x{:}A_m)\, \Omega_R)(x{:}A_m)^{n-1} \vdash C_r}{\mathcal{E}'}
      & (2)}
    & (1)}
\end{smallrules}
With 
\[
   (1) = \Omega_A \geq m \geq r\ \ \  (2) = \lcontr \in \sigma(m)
\]

\begin{tabbing}
  $(\Omega_L\,\Omega_A\,\Omega_M\,\Omega_A\,\Omega_R)(\Omega_A)^{n-1} \vdash C_r$ \`
  by IH on $(A_m,\mathcal{D},\mathcal{E}')$, with (1)\\
  $\Omega_A$ admits left contraction \` by monotonicity $((2),(1))$\\
  $(\Omega_L\, \Omega_A\, \Omega_M\, \Omega_R)(\Omega_A)^{n-1} \vdash C_r$ \` by $\lcontr$ on elements of $\Omega_A$
\end{tabbing}
\item[Case: $\mathcal{E}$ ends in $\lmove$ on (one copy of) principal formula]
\begin{smallrules}
  \infer-[\mathsf{mcut}]
  {(\Omega_L\, \Omega_A\, \Omega_M\, \Omega_R)(\Omega_A)^{n-1} \vdash C_r}
  {\deduce{\Omega_A \vdash A_m}{\mathcal{D}}
    & \infer[\lmove]
    {(\Omega_L\, (x{:}A_m)\, \Omega_M\, \Omega_R)(x{:}A_m)^{n-1} \vdash C_r}
    {\deduce{(\Omega_L\, \Omega_M\, (x{:}A_m)\, \Omega_R)(x{:}A_m)^{n-1} \vdash C_r}{\mathcal{E}'} & (2)}
    & (1)}
\end{smallrules}
With 
\[(1) = \Omega_A \geq m \geq r\]
\[(2) = \lmove \in \sigma(m)\] 
\begin{tabbing}
  $(\Omega_L\, \Omega_M\, \Omega_A\, \Omega_R)(\Omega_A)^{n-1} \vdash C_r$ \` by IH on $(A_m,\mathcal{D},\mathcal{E}')$, with (1)\\
  $\Omega_A$ admits left mobility \` by monotonicity $((2),(1))$\\
  $(\Omega_L\, \Omega_A\, \Omega_M\, \Omega_R)(\Omega_A)^{n-1} \vdash C_r$ \` by $\lmove$ on elements of $\Omega_A$, with previous
\end{tabbing}
\end{description}   
\end{proof}

\begin{theorem}[Cut Elimination]
  \label{thm:cut-elim}
  If $\Omega \vdash C_r$ then there is a cut-free proof of this sequent.
\end{theorem}
\begin{proof}
  By straightforward structural induction on the given derivation.  For the case
  of cut, we apply admissibility of multicut to the result of eliminating cut
  from the derivations of the premises.
\end{proof}

As a second test for the system, the identity rule should be admissible except
for atoms.  As usual, this is straightforward in a sequent calculus.

\begin{theorem}[Admissibility of Identity]
  \label{thm:id-admit}
  The rule of identity is admissible in a system where it is restricted to
  atomic propositions.
\end{theorem}
\begin{proof}
  By induction over the structure of the proposition $A_m$.
\end{proof}

\section{Natural Deduction}
\label{sec:nd}
We now move towards a presentation of a natural deduction system in which
structural rules are implicit. We do so in a natural deduction context rather
than staying within the sequent calculus as the final goal of this work is to
have a type checker for an extension of the adjoint natural deduction system in
\cite{Jang24fscd} that considers order. However, for the initial presentation
we elide proof terms and show the decidability of natural deduction proof
checking in the absence of explicit structural rules instead.  The final
step to bidirectional type-checking is then rather straightforward and
is patterned after \cite{Jang24fscd}. 
A short description of this process can be found in \autoref{sec:types}.

It is perhaps somewhat unexpected that proof checking with implicit structural
rules is not easily seen to be decidable, as is the case in unordered
adjoint natural deduction. The culprits are the interactions of unidirectional
mobility, weakening, and the nondeterminism in locating the new antecedents in
the positive elimination rules.  The corresponding problem in the sequent
calculus is that in the cut rule, the new antecedent $x{:} A_m$ can be anywhere
in the context, a property that is required for cut elimination.

We begin with a non-algorithmic system in \autoref{sec:nd-explicit}, which
still keeps structural rules explicit to build a bridge to the sequent calculus.
In \autoref{sec:nd-implicit} we present a terminating algorithm to check
proofs with implicit structural rules.  Optimizations beyond the scope of this
paper are still necessary to make this algorithm practical and are briefly
discussed in the conclusion.

\subsection{Explicit Structural Rules}
\label{sec:nd-explicit}

We begin with the judgment 
\[
  \Gamma \vdash A_m \dashv \Omega
\]
where $\Gamma$ is an unordered set of all hypotheses $y{:} B_k$ that are
syntactically in scope while proving $A_m$. $\Omega$ is the context of the
ones that are actually used, in order. 

One consideration with such a system comes from the nullary constructors $\oplus_m\{\}$ and $\with\{\}$. 
The right rule for the nullary $\with_m\{\}$ allows any of the input context to be considered as used as long 
as it is at the appropriate mode. 
To address this \cite{Jang24fscd} used provisional bindings to indicate in the output 
that certain propositions could be considered used but don't have to be. 
This solution should extend in a straightforward manner to the ordered setting, and so we consider it orthogonal to the main contributions of this paper. 
For this reason, in this presentation, we consider only the binary $\oplus$ and $\with$ to avoid the extra complication of provisional bindings.

As usual in natural deduction, rules are divided into introduction and elimination
rules.  We consider $\fuse I$ as a first sample rule.
\begin{smallrules}
  \infer[\fuse I]
  {\Gamma \vdash A_m \fuse B_m \dashv \Omega_A\, \Omega_B}
  {\Gamma \vdash A_m \dashv \Omega_A & \Gamma \vdash B_m \dashv \Omega_B}
\end{smallrules}
$\Omega_A$ and $\Omega_B$ are the ordered lists of labeled hypothesis used in
the proofs of $A_m$ and $B_m$, respectively.  They must be concatenated to form
the hypotheses used to prove the pair.

In the presence of contraction, the same labeled hypothesis
$y{:} B_k$ may be present in both $\Omega_A$ and $\Omega_B$.  In unordered
adjoint (and other substructural logics), they can be directly merged, but
here we require an explicit application of a contraction rule, if permissible.
\begin{smallrules}
  \infer[\rcontr]
  {\Gamma \vdash  C_r \dashv \Omega_L\, \Omega_M\, (x{:}A_m)\, \Omega_R}
  {\Gamma \vdash  C_r \dashv \Omega_L\, (x{:}A_m)\, \Omega_M\, (x{:}A_m)\, \Omega_R
    & \mathsf{C}^{\rightarrow} \in \sigma(m)}
  \\[0.75em]
  \infer[\lcontr]
  {\Gamma \vdash  C_r \dashv \Omega_L\, (x{:}A_m)\, \Omega_M\, \Omega_R}
  {\Gamma \vdash  C_r \dashv \Omega_L\, (x{:}A_m)\, \Omega_M\, (x{:}A_m)\, \Omega_R
    & \mathsf{C}^{\leftarrow} \in \sigma(m)}
\end{smallrules}
The $\fuse E$ rule is one source of considerable nondeterminism.  We write
$|\Omega|$ for the set of variables declared in $\Omega$.
\begin{smallrules}
  \infer[\fuse E]
  {\Gamma \vdash C_r \dashv \Omega_L\, \Omega\, \Omega_R}
  {\Gamma \vdash A_m \fuse B_m \dashv \Omega
    & (m \geq r)
    & \begin{array}[b]{r}
        (x, y \not\in |\Omega_L\, \Omega_R|) \\
        \Gamma, x {:} A_m, y {:} B_m \vdash C_r \dashv \Omega_L\, (x {:} A_m)\, (y {:} B_m)\, \Omega_R
      \end{array}}
\end{smallrules}

\noindent One difficulty is that due to unidirectional mobility and also due to weakening,
the location for $x$ and $y$ in the output context of the second premise may not
be uniquely determined.  Furthermore, $x$ and $y$ go out of scope, so they are
not allowed to occur in $\Omega_L$ and $\Omega_R$ (nor in $\Omega$ by our
general freshness assumption).

We also have rules where two output contexts must match. $\with I$ is shown:
\begin{smallrules}
  \infer[\with I]
  {\Gamma \vdash A_m \with B_m \dashv\Omega}
  {\Gamma \vdash  A_m \dashv \Omega & \Gamma \vdash B_m \dashv \Omega}
\end{smallrules}
This might require that some weakening occur before application of this
rule in the first and second premise (assuming the mode $m$ permits it).

The structural rules for weakening and mobility on the output context are:
\begin{smallrules}
  \infer[\ms{W}]
  {\Gamma \vdash C_r \dashv \Omega_L(x{:}A_m)\Omega_R}
  {\Gamma \vdash  C_r \dashv \Omega_L\Omega_R & x{:}A_m \in \Gamma & (m \geq r) & \mathsf{W} \in \sigma(m)}
  \\[0.75em]
  \infer[\rmove]
  {\Gamma \vdash  C_r \dashv \Omega_L\, \Omega_M\, (x{:}A_m)\, \Omega_R}
  {\Gamma \vdash  C_r \dashv \Omega_L\, (x{:}A_m)\, \Omega_M\, \Omega_R
    & \mathsf{M}^{\rightarrow} \in \sigma(m)}
  \\[0.75em]
  \infer[\lmove]
  {\Gamma \vdash C_r \dashv \Omega_L\, (x{:}A_m)\, \Omega_M\, \Omega_R}
  {\Gamma \vdash C_r \dashv \Omega_L\, \Omega_M\, (x{:}A_m)\, \Omega_R
    & \mathsf{M}^{\leftarrow} \in \sigma(m)}       
\end{smallrules}
The remaining logical rules can be found in \autoref{fig:ordnd}.
\begin{figure}
  \begin{smallrules}
    \infer[1I]{\Gamma \vdash 1_m \dashv \cdot}{} \hspace{3em}
    \infer[1E]{\Gamma \vdash C_r \dashv \Omega_L\, \Omega_M\, \Omega_R}
    {\Gamma \vdash1_m\dashv {\Omega_M} & (m \geq r)
      & \Gamma \vdash C_r \dashv {\Omega_L\, \Omega_R}}
    \\[0.75em]
    \infer[\fuse I]{\Gamma \vdash A_m \fuse B_m \dashv \Omega_L\Omega_R}
    {\Gamma \vdash A_m \dashv {\Omega_L} & \Gamma \vdash B_m \dashv {\Omega_R}}
    \\[0.75em]
    \infer[\fuse E]
    {\Gamma \vdash C_r\dashv {\Omega_L\, \Omega_M\, \Omega_R}}
    {\Gamma \vdash A_m \fuse B_m\dashv {\Omega_M}
      & (m \geq r)
      & \begin{array}[b]{r}
          (x,y \notin \Omega_L\, \Omega_R) \\
          \Gamma, x{:}A_m, y{:}B_m \vdash C_r \dashv {\Omega_L\, (x{:}A_m)\, (y{:}B_m)\, \Omega_R}
        \end{array}}
    \\[0.75em]
        
    \infer[\oplus I_1]
    {\Gamma \vdash A_m \oplus B_m\dashv {\Omega}}
    {\Gamma \vdash A_m\dashv {\Omega}}
    \hspace{3em}
    \infer[\oplus I_2]
    {\Gamma \vdash A_m \oplus B_m\dashv \Omega}
    {\Gamma \vdash B_m\dashv {\Omega}}
    \\[0.75em]
    \infer[\oplus E]
    {\Gamma \vdash C_r\dashv \Omega_L\, \Omega_M\, \Omega_R}
    {\Gamma \vdash A_m \oplus B_m \dashv {\Omega_M}
      & (m \geq r)
      & \begin{array}[b]{r}
          x \notin |\Omega_L\,\Omega_R| \\
          \Gamma, x{:}A_m \vdash C_r \dashv {\Omega_L\, (x{:}A_m)\, \Omega_R} \\
          \Gamma, x{:}B_m \vdash C_r \dashv {\Omega_L\, (x{:}B_m)\, \Omega_R}
        \end{array}}
  \\[0.75em]
  \infer[\down I]
  {\Gamma \vdash \down^k_m A_k \dashv {\Omega}}
  {\Gamma \vdash A_k \dashv {\Omega} }
  \\[0.75em]
  \infer[\down E]
  {\Gamma \vdash C_r \dashv \Omega_L\, \Omega_M\, \Omega_R}
  {\Gamma \vdash \down^k_m A_k \dashv {\Omega_M}
    & (m \geq r)
    & \Gamma, x{:}A_k \vdash C_r \dashv {\Omega_L\, (x{:}A_k)\, \Omega_R}
    & x \notin |\Omega_L\,\Omega_R|}
  \\[0.75em]\hline\\[0.75em]
  \infer[\rightimp I]
  {\Gamma \vdash A_m \rightimp B_m\dashv \Omega}
  {\Gamma,x{:}A_m \vdash B_m \dashv {\Omega\, (x{:}A_m)}
    & x \notin \Omega}
  \hspace{1em}
  \infer[\rightimp E]
  {\Gamma \vdash  B_m\dashv \Omega_L\Omega_R}
  {\Gamma \vdash A_m \rightimp B_m\dashv {\Omega_L}
    & \Gamma \vdash A_m\dashv {\Omega_R}}
  \\[0.75em]
  \infer[\leftimp I]
  {\Gamma \vdash A_m \leftimp B_m \dashv \Omega}
  {\Gamma, x{:}A_m \vdash B_m \dashv {(x{:}A_m)\, \Omega}
    & x \notin \Omega}
  \hspace{1em}
  \infer[\leftimp E]
  {\Gamma \vdash B_m \dashv {\Omega_L\, \Omega_R}}
  {\Gamma \vdash A_m \leftimp B_m \dashv {\Omega_R}
    & \Gamma \vdash A_m \dashv {\Omega_L}}
  \\[0.75em]
  \infer[\with I]
  {\Gamma \vdash A_m \with B_m \dashv {\Omega}}
  {\Gamma \vdash A_m \dashv {\Omega}
    & \Gamma \vdash B_m \dashv {\Omega}}
  \hspace{2em}
  \infer[\with E_1]
  {\Gamma \vdash A_m \dashv {\Omega}}
  {\Gamma \vdash A_m \with B_m \dashv {\Omega}}
  \hspace{1em}
  \infer[\with E_2]
  {\Gamma \vdash B_m \dashv {\Omega}}
  {\Gamma \vdash  A_m \with B_m \dashv {\Omega}}
  \\[0.75em]
  \infer[\up I]
  {\Gamma \vdash \up_l^m A_l \dashv {\Omega}}
  {\Gamma \vdash A_l \dashv {\Omega}
    & \Omega \geq m }
  \hspace{3em}
  \infer[\up E]
  {\Gamma \vdash A_k \dashv {\Omega}}
  {\Gamma \vdash \up_k^l A_k \dashv {\Omega}}
  \\[0.75em]\hline\\[0.75em]
  \infer[\mathsf{hyp}]
  {\Gamma \vdash A_m \dashv (x{:}A_m)}
  {x{:}A_m \in \Gamma}
\end{smallrules}
\caption{Ordered Natural Deduction (explicit structural rules in \autoref{sec:nd-explicit})}
\label{fig:ordnd}
\end{figure}

We relate this system to the sequent calculus defined in
\autoref{sec:ordered-seq}.  There are more fine-grained translations (say,
relating cut-free sequent derivations to normal natural deduction; see
\cite{Jang24fscd} for the unordered adjoint case), but this is not relevant for
our purposes here.
\begin{lemma}[Substitution]
  \begin{enumerate}
  \item If $\Gamma \vdash A_m \dashv \Omega$ then $\Gamma \supseteq \Omega$ and $\Omega \geq m$
  \item If $\Gamma \vdash A_m \dashv \Omega$ and $\Gamma' \supseteq \Gamma$ then $\Gamma' \vdash A_m \dashv \Omega$
  \item If $\Gamma_1 \vdash A_m \dashv \Omega_A$ and $\Gamma_2, x{:} A_m \vdash C_r \dashv \Omega(x{:} A_m)^n$ \newline
    then $\Gamma_1 \cup \Gamma_2 \vdash C_r \dashv \Omega(\Omega_A)^n$ provided $|m| \sim n$.
  \end{enumerate}
\end{lemma}
\begin{proof}
  All by simple rule inductions on given derivations.  
\end{proof}

\begin{theorem}[Sequent Calculus to Natural Deduction]
  \label{thm:seq2nd}\mbox{}\newline
  If $\Omega \vdash A_m$ then $\Gamma \vdash A_m \dashv \Omega$ for all $\Gamma \supseteq \Omega$
\end{theorem}
\begin{proof}
  By induction on the given derivation, using substitution in several cases. 
\end{proof}

\begin{theorem}[Natural Deduction to Sequent Calculus]
  \label{thm:nd2seq}\mbox{}\newline
  If $\Gamma \vdash A_m \dashv \Omega$ then $\Omega \vdash A_m$.
\end{theorem}
\begin{proof}
  By induction on the structure of the given derivation, using cut and identity
  in several cases.
\end{proof}

\subsection{Implicit Structural Rules}
\label{sec:nd-implicit}
If these explicit structural rules from the last section are elided, checking the validity of the
remaining proof skeleton is not immediately decidable.  One issue arises
from the rule of weakening
\begin{smallrules}
  \infer[\mathsf{W}]
  {\Gamma \vdash  C_r \dashv \Omega_L\, (x{:}A_m)\, \Omega_R}
  {\Gamma \vdash C_r\dashv \Omega_L\, \Omega_R
    & (m \geq r) & x{:}A \in \Gamma & \mathsf{W} \in \sigma(m) }
\end{smallrules}
which could be applied arbitrarily often.  To address this issue, we take
advantage of the following property, observed by
\cite{Kanovich19mscs}: mobility is derivable from weakening and contraction.
The following shows how to derive left mobility from left contraction and weakening:
\begin{smallrules}
  \infer[\lcontr]
  {\Gamma \vdash C_r \dashv {\Omega_L\, (x{:}A_m)\, \Omega_M\, \Omega_R}}
  {\infer[\mathsf{W}]
    {\Gamma \vdash C_r \dashv {\Omega_L\, (x{:}A_m)\, \Omega_M\, (x{:}A_m)\, \Omega_R}}
    {\Gamma \vdash C_r \dashv {\Omega_L\, \Omega_M\, (x{:}A_m)\, \Omega_R}
        & (m \geq r) & (x{:}A) \in \Gamma
        & \mathsf{W} \in \sigma(m)}
      & \lcontr \in \sigma(m)}
\end{smallrules}
While we could just use this derived rule, it is convenient to assume the
following two closure properties:
\begin{align}
  \text{If }\mathsf{W} \in \sigma(m) \text{ and } \mathsf{C}^\leftarrow \in \sigma(m) \text{ then } \ms{M}^\leftarrow \in \sigma(m)\\
  \text{If }\mathsf{W} \in \sigma(m) \text{ and } \mathsf{C}^\rightarrow \in \sigma(m) \text{ then } \ms{M}^\rightarrow \in \sigma(m)
  \end{align}

Under conditions (1) and (2), we can restrict weakening
to only be allowed when the variable is not yet present in the output as
follows:
\begin{smallrules}
  \infer[\mathsf{W}]
  {\Gamma \vdash C_r \dashv {\Omega_L\, (x{:}A_m)\, \Omega_R}}
  {\Gamma \vdash C_r \dashv {\Omega_L\, \Omega_R}
    & (m \geq r)
    & (x{:} A_m) \in \Gamma
    & \mathsf{W} \in \sigma(m)
    & x \not\in |\Omega_L\, \Omega_R|}
\end{smallrules}
The following also hold, however these are not explicitly needed anywhere, but they reduce the total number of unique combinations of structural properties: 
\begin{align}
    \text{If }\lcontr \in \sigma(m) \text{ and } \rmove \in \sigma(m) \text{ then } \rcontr \in \sigma(m)\\
    \text{If }\rcontr \in \sigma(m) \text{ and } \lmove \in \sigma(m) \text{ then } \lcontr \in \sigma(m)
      \end{align}

\noindent To construct a system with implicit structural rules that has a
decidable algorithm for proof checking, we bundle the structural rules into a
new judgment $\Omega \reduces{\Gamma}{r} \OmegaNF$ where $\OmegaNF$ denotes a
\emph{normal form} where no variable is repeated.  We write $\Omega\; \nf$ for
the judgment that holds iff $\Omega$ is normal in this sense.  We carry the
unordered $\Gamma$ and mode $r$ in order to control weakening, with the
presupposition that $\Gamma \supseteq \Omega$ and $\Omega \geq r$ and ensure
that $\Gamma \supseteq \OmegaNF$ and $\OmegaNF \geq r$.  Algorithmically, we
view the left-hand side $\Omega$ together with $\Gamma$ and $r$ as inputs and
$\OmegaNF$ as output.
\begin{smallrules}
  \infer[\mathsf{W}]{\Omega_L\Omega_R \reduces{\Gamma}{r} \OmegaNF}
  {\Omega_L(x{:}A_m)\Omega_R \reduces{\Gamma}{r} \Omega_{\mathsf{nf}}
    & \mathsf{W} \in \sigma(m)
    & m \geq r
    & x{:}A_m \in \Gamma & x \notin |\Omega_L\Omega_R|}
  \\[0.75em]
  \infer[\mathsf{\lcontr}]
  {\Omega_L(x{:}A_m)\Omega_M(x{:}A_m)\Omega_R \reduces{\Gamma}{r} \Omega_{\mathsf{nf}}}
  {\Omega_L(x{:}A_m)\Omega_M\Omega_R \reduces{\Gamma}{r} \Omega_{\mathsf{nf}} & \lcontr \in \sigma(m)}
  \\[0.75em]
  \infer[\mathsf{\rcontr}]
  {\Omega_L(x{:}A_m)\Omega_M(x{:}A_m)\Omega_R \reduces{\Gamma}{r} \Omega_{\mathsf{nf}}}
  {\Omega_L\Omega_M(x{:}A_m)\Omega_R \reduces{\Gamma}{r} \Omega_{\mathsf{nf}} & \rcontr \in \sigma(m)}
  \\[0.75em]
  \infer[\mathsf{\lmove}]
  {\Omega_L\Omega_M(x{:}A_m)\Omega_R \reduces{\Gamma}{r} \Omega_{\mathsf{nf}}}
  {\Omega_L(x{:}A_m)\Omega_M\Omega_R \reduces{\Gamma}{r} \Omega_{\mathsf{nf}} & \lmove \in \sigma(m)}
  \\[0.75em]
  \infer[\mathsf{\rmove}]
  {\Omega_L(x{:}A_m)\Omega_M\Omega_R \reduces{\Gamma}{r} \Omega_{\mathsf{nf}}}
  {\Omega_L\Omega_M(x{:}A_m)\Omega_R \reduces{\Gamma}{r} \Omega_{\mathsf{nf}} & \rmove \in \sigma(m)}
  \\[0.75em]
  \infer[\ms{id}]{\Omega \reduces{\Gamma}{r} \Omega}{(\Omega\; \nf)}
\end{smallrules}
For any given $\Omega$, $\Gamma$, and $r$, there may be many $\OmegaNF$ such
that $\Omega \reduces{\Gamma}{r} \OmegaNF$.  To account for this we define
\[
  \NF_{\Gamma, r}(\Omega) = \{\OmegaNF \mid \Omega \reduces{\Gamma}{r} \OmegaNF\}
\]
For example, given $\Gamma = x{:}A_m,y{:}B_k$, $\Omega = y{:}B_k$ with $m \geq r$, $k \geq r$, and $\sigma(m) = \{\mathsf{W}\}$ we obtain
\(
  \NF_{\Gamma,r}(\Omega) = \{(y:B_k), (x:A_m)(y:B_k), (y:B_k)(x:A_m)\}
  \)
  
Since $\Gamma$ and the mode $r$ are usually immediately apparent from context,
we often omit them.  Fortunately, $\NF_{\Gamma, r}(\Omega)$ is always finite because the
only rule that increases the size of the context (reading bottom-up) is
weakening, and it is restricted to variables that are in $\Gamma$ but not yet in
the input context $\Omega$.


We define a new judgment
$\Gamma \vdash A_m \Dashv \Xi$, where $\Xi$ is a \emph{set of ordered contexts} in normal form.  The guiding principle is embodied in
\autoref{thm:complete-implicit}: If $\Gamma \vdash A_m \dashv \Omega$ then
$\Gamma \vdash A_m \Dashv \Xi$ where all normal forms of $\Omega$ are contained
in $\Xi$.

Logical rules are defined via set construction. Because the output of the judgment $\Gamma \vdash A_m \Dashv \Xi$
is a set of contexts, we lift the normal form operation to sets, defining
\[
  \NF_{\Gamma, r}(\Xi) = \bigcup_{\Omega \in \Xi}\{ \Omega' \mid \Omega' \in \NF_{\Gamma, r}(\Omega)\}
\]
Consider the $\fuse I$ rule (We write $\Xi_1\, \Xi_2$ for $\{\Omega_1\, \Omega_2 \mid \Omega_1 \in \Xi_1, \Omega_2 \in \Xi_2\}$):
\begin{smallrules}
  \infer[\fuse I]
  {\Gamma \vdash A_m \fuse B_m \Dashv \NF_{\Gamma, m}(\Xi_1\, \Xi_2)}
  {\Gamma \vdash  A_m \Dashv \Xi_1 & \Gamma \vdash B_m \Dashv \Xi_2 
    }
\end{smallrules}

\noindent The $\rightimp I$ rule becomes:
\begin{smallrules}
  \infer[\rightimp I]
  {\Gamma \vdash A_m \rightimp B_m \Dashv \Xi}
  {\Gamma,x{:}A_m \vdash  B_m \Dashv \Xi' & \Xi'\bbar\Xi(x{:}A_m)}
\end{smallrules}
$\Xi' \bbar \Xi(x{:}A_m)$ filters $\Xi'$ to only contexts that end in
$x{:}A_m$. Then in the conclusion
$\Xi = \{ \Omega \mid \Omega\, (x {:} A_m) \in \Xi'\}$.  Similarly, we write
$\Xi \bbar_m = \{\Omega \in \Xi \mid \Omega \geq m\}$.  This is used in the $\up I$ rule. Implicitly, all rules fail if no output contexts are possible. 

\noindent
The last operator is related to the case where multiple output contexts must be equivalent, that is the $\with I$ rule and the $\oplus E$ rule. We walk through $\with I$: 
\begin{smallrules}
    \infer[\with I]{\Gamma \vdash A_m \with B_m \Dashv \Xi_1\cap \Xi_2}{\Gamma \vdash  A_m \Dashv \Xi_1 & \Gamma \vdash  B_m \Dashv \Xi_2 
    }
\end{smallrules}
Since $\Xi_1$ and $\Xi_2$ might not be the same, we
output the intersection to capture those contexts on which they do agree.  The other rules can be found in \autoref{fig:implicit-nd}.

As an example of this system, consider the following derivations (we use just the variable name on the output context for space). 
We place each context in the set in square brackets as a visual separator.
\begin{scriptsize}
  \begin{smallrules}
    \infer[\fuse E]{x{:}\up_m^k (A_m \with B_m) \vdash \up_m^k A_m \fuse \up_m^k B_m \Dashv {NF}(\{[x]\}\{[x]\}) }
    {\infer[\up I]
      {x{:}\up_m^k (A_m \with B_m) \vdash \up_m^kA_m \Dashv \{[x]\}}
      {\infer[\with E]
        {x{:}\up_m^k (A_m \with B_m) \vdash A_m \Dashv \{[x]\}}
        {\infer[\up E]
          {x{:}\up_m^k (A_m \with B_m) \vdash A_m \with B_m \Dashv \{[x]\}}
          {\infer[\mathsf{hyp}]
            {x{:}\up_m^k (A_m \with B_m) \vdash \up_m^k A_m \with B_m \Dashv \{[x]\}}
            {}
            & k \geq k}}}
      & \hspace*{-3em}\infer[\up I]{x{:}\up_m^k (A_m \with B_m) \vdash \up_m^kB_m \Dashv\{[x]\}}{\cdots}}
\end{smallrules}  
\end{scriptsize}%
For the NF operation to succeed in the conclusion, the mode $k$ must admit contraction in either direction. The reverse entailment requires that $k$ admit weakening.



\begin{figure}
  \begin{smallrules}
    \infer[1I]{\Gamma \vdash 1_m \Dashv \NF(\{\cdot\})}{} \hspace{3em}
    \infer[1E]{\Gamma 
      \vdash C_r \Dashv \NF(\Xi_L\, \Xi_M\, \Xi_R)}{\Gamma \vdash  1_m\Dashv {\Xi_M} & m \geq r & \Gamma 
      \vdash C_r \Dashv {\Xi_L\, \Xi_R}}
    \\[0.75em]
    \infer[\fuse I]{\Gamma 
      \vdash A_m \fuse B_m \Dashv \NF(\Xi_L\Xi_R)}{\Gamma 
      \vdash A_m \Dashv {\Xi_L} & \Gamma 
      \vdash B_m \Dashv {\Xi_R}}
    \\[0.75em]
    \infer[\fuse E]{\Gamma 
      \vdash C_r\Dashv \NF(\Xi_L\Xi_M\Xi_R)}{\Gamma \vdash  A_m \fuse B_m\Dashv {\Xi_M} & m \geq r & \Gamma,x{:}A_m,y{:}B_m 
      \vdash C_r \Dashv \Xi & \Xi\bbar\Xi_L(x{:}A_m) (y{:}B_m) \Xi_R}
    \\[0.75em]
    \infer[\oplus I_1]{\Gamma 
      \vdash A_m \oplus B_m \Dashv \Xi}{\Gamma 
      \vdash A_m \Dashv {\Xi }}\hspace{3em}
    \infer[\oplus I_2]{\Gamma 
      \vdash A_m \oplus B_m \Dashv \Xi}{\Gamma 
      \vdash B_m \Dashv {\Xi }}
    \\[0.75em]
    \infer[\oplus E]{\Gamma 
      \vdash C_r\Dashv \NF(\Xi_L\Xi_M\Xi_R)}{\Gamma \vdash A_m \oplus B_m\Dashv \Xi_M  & m \geq r & \deduce{\Gamma,x{:}B_m \vdash C_r \Dashv \Xi_B}{ \Gamma,x{:}A_m \vdash C_r \Dashv \Xi_A} & 
      \deduce{\Xi_L\Xi_R = \Xi_L^A\Xi_R^A \cap \Xi_L^B\Xi_R^B}{\deduce{\Xi_B \bbar \Xi_L^B\, (x{:}B_m)\, \Xi_R^B}{\Xi_A \bbar \Xi_L^A\, (x{:}A_m)\, \Xi_R^A}} }
    \\[0.75em]
    \infer[\down I]{\Gamma \vdash  \down^l_m A_l \Dashv \NF_{\Gamma,m}(\Xi)}{\Gamma 
      \vdash A_l \Dashv {\Xi} }
    \\[0.75em]
    \infer[\down E]
    {\Gamma \vdash C_r \Dashv \NF_{\Gamma, r}(\Xi_L\Xi_M\Xi_R)}
    {\Gamma \vdash  \down^l_m A_l \Dashv {\Xi_M} & m \geq r & \Gamma,x{:}A_l 
      \vdash C_r\Dashv \Xi & \Xi\bbar {\Xi_L\, (x{:}A_l)\, \Xi_R}}
    \\[0.75em]\hline\\[0.75em]
    \infer[\rightimp I]{\Gamma 
      \vdash A_m \rightimp B_m\Dashv \Xi}{\Gamma,x{:}A_m 
      \vdash B_m \Dashv \Xi'  & \Xi' \bbar \Xi(x{:}A_m)}
    \hspace{3em}
    \infer[\rightimp E]{\Gamma \vdash  B_m\Dashv \NF(\Xi_L\Xi_R)}{\Gamma \vdash  A_m \rightimp B_m\Dashv {\Xi_L} & \Gamma 
      \vdash A_m\Dashv {\Xi_R}}
    \\[0.75em]
    \infer[\leftimp I]{\Gamma 
      \vdash A_m \leftimp B_m \Dashv \Xi}{\Gamma,x{:}A_m 
      \vdash B_m \Dashv {\Xi'} & \Xi' \bbar (x{:}A_m)\Xi}
    \hspace{3em}
    \infer[\leftimp E]{\Gamma \vdash  B_m\Dashv \NF(\Xi_L\Xi_R)}{\Gamma \vdash  A_m \leftimp B_m \Dashv {\Xi_R} & \Gamma 
      \vdash A_m \Dashv {\Xi_L}}
    \\[0.75em]
    \infer[\with I]{\Gamma 
      \vdash A_m \with B_m \Dashv \Xi_A \cap \Xi_B}{\Gamma 
      \vdash A_m \Dashv \Xi_A & \Gamma \vdash B_m \Dashv \Xi_B}
    \\[0.75em]
    \infer[\with E_1]{\Gamma \vdash A_m \Dashv \Xi}{\Gamma \vdash  A_m \with B_m \Dashv \Xi}
    \hspace{3em}
    \infer[\with E_2]{\Gamma \vdash B_m \Dashv \Xi}{\Gamma \vdash  A_m \with B_m \Dashv \Xi}
    \\[0.75em]
    \infer[\up I]{\Gamma \vdash \up_k^m A_k\Dashv \Xi\bbar_m}{\Gamma 
      \vdash A_k\Dashv {\Xi}}
    \hspace{3em}
    \infer[\up E]{\Gamma \vdash A_k \Dashv \NF_{\Gamma,k}(\Xi)}{\Gamma \vdash \up_k^m A_k \Dashv {\Xi}}
    \\[0.75em]\hline\\[0.75ex]
    \infer[\mathsf{hyp}]{\Gamma \vdash  A_m \Dashv \NF(\{x{:}A_m\})}{x{:}A_m \in \Gamma}
  \end{smallrules}
  \caption{Implicit Natural Deduction}
  \label{fig:implicit-nd}
\end{figure}

\subsection{Proof of Correctness and Termination}
\label{sec:decidability}

We prove the system with implicit structural rules sound and complete with
respect to the system defined in \autoref{sec:nd-explicit}.
We start with a straightforward lemma.
 \begin{lemma}[Contraction commutes with other structural rules]
 \label{lemma:contraction} 
    \end{lemma}
    \begin{proof}
    Proof proceeds by a series of reductions.
    \begin{description}
    \item[Case: Mobility.]
    We show contraction followed by mobility in the same direction. Here the dotted line indicates a vacuous rule. We could also reduce this to just a single contraction.
    \begin{smallrules}
    \infer[M^\rightarrow]{\Omega_{L1}(x{:}A_m)\Omega_{L2}\Omega_M(x{:}A_m)\Omega_R \reduces{\Gamma}{r} \Omega}{\infer[C^\rightarrow]{\Omega_{L1}\Omega_{L2}(x{:}A_m)\Omega_M(x{:}A_m)\Omega_R \reduces{\Gamma}{r} \Omega}{\Omega_{L1}\Omega_{L2}\Omega_M(x{:}A_m)\Omega_R \reduces{\Gamma}{r} \Omega & C^\rightarrow \in \sigma(m)} & M^\rightarrow \in \sigma(m)}
\\
\rightsquigarrow
\\
    \infer[C^\rightarrow]{\Omega_{L1}(x{:}A_m)\Omega_{L2}\Omega_M(x{:}A_m)\Omega_R \reduces{\Gamma}{r} \Omega}{\infer-[M^\rightarrow]{\Omega_{L1}\Omega_{L2}\Omega_M(x{:}A_m)\Omega_R \reduces{\Gamma}{r} \Omega}{\Omega_{L1}\Omega_{L2}\Omega_M(x{:}A_m)\Omega_R \reduces{\Gamma}{r} \Omega & M^\rightarrow \in \sigma(m)} & C^\rightarrow \in \sigma(m)}
\end{smallrules}
    
\end{description}
        
   \end{proof}
With this lemma, we can proceed to state and prove the following theorem:

\begin{theorem}[Restricted Weakening is Complete]
\label{thm:complete-restriction}
A version of context reduction where the added hypothesis may already exist
produces the same set of normal form as the one with restricted weakening.
\end{theorem}
\begin{proof}
  One inclusion is immediate.  The other permutes contractions
  with other structural rules, replacing weakening followed by
  contraction with mobility.
\begin{smallrules}
    \infer[\mathsf{W}]{\Omega_L(x{:}A_m)\Omega_M\Omega_R \reduces{\Gamma}{r} \Omega}{\deduce{\Omega_L(x{:}A_m) \Omega_M (x{:}A_m) \Omega_R \reduces{\Gamma}{r} \Omega}{\mathcal{D}} & x \in \Gamma & \mathsf{W} \in \sigma(m) & m \geq r}
\end{smallrules}
By IH($\mathcal{D}$) we conclude:
    \[\deduce{\Omega_L(x{:}A_m)\Omega_M(x{:}A_m)\Omega_R \reduces{\Gamma}{r} \Omega}{\mathcal{D'}}\]
    
Because we can reduce to $\Omega$ which must be in normal form, there must exist a contraction of $x{:}A_m$,  in the derivation $\mathcal{D'}$. By lemma \ref{lemma:contraction}, we push this contraction to the end, and have a valid derivation of the premise. And so we have the following (modified) derivation of the conclusion of $\mathcal{D}'$ 
\begin{smallrules}
    \infer[\lcontr]{\Omega_L(x{:}A_m)\Omega_M(x{:}A_m)\Omega_R \reduces{\Gamma}{r} \Omega}{\deduce{\Omega_L(x{:}A_m)\Omega_M\Omega_R \reduces{\Gamma}{r} \Omega}{\mathcal{D}'' } & \lcontr \in \sigma(m)}
\end{smallrules}
The derivation $\mathcal{D}''$ then has the conclusion that we require. 

Right contraction requires an extra step, we begin by applying the inductive hypothesis, and pushing down the right contraction, as with the other direction, resulting in: 
\begin{smallrules}
    \infer[\rcontr]{\Omega_L(x{:}A_m)\Omega_M(x{:}A_m)\Omega_R \reduces{\Gamma}{r} \Omega}{\deduce{\Omega_L\Omega_M(x{:}A_m)\Omega_R \reduces{\Gamma}{r} \Omega}{\mathcal{D}'''} & \rcontr \in \sigma(m)}
\end{smallrules}
At this point, we use the property that contraction plus weakening result in mobility in the same direction as the contraction. We can now apply the mobility right rule to the conclusion of derivation $\mathcal{D}'''$ as follows: 
\begin{smallrules}
    \infer[\rmove]{\Omega_L(x{:}A_m)\Omega_M\Omega_R \reduces{\Gamma}{r} \Omega}{\deduce{\Omega_L\Omega_M(x{:}A_m)\Omega_R \reduces{\Gamma}{r} \Omega}{\mathcal{D}'''} & \rmove \in \sigma(m)}
\end{smallrules}
\end{proof}

\begin{theorem}
  \label{thm:termination}
  Construction of the set $\NF_{\Gamma,r}(\Omega) = \{\Omega'\mid \Omega \reduces{\Gamma}{r}\Omega'\}$ terminates.
\end{theorem}
\begin{proof}\mbox{}
  Weakening is bounded by $\Gamma$, while contraction is bounded by the number
  of duplicate variables (which do not increase by restricted weakening).
\end{proof}   

\noindent We need one more lemma to finally prove correctness:

\begin{lemma}[Concatenation]
  \label{lemma:join-ok}\mbox{}\newline
  $\forall \Omega \in \NF_{\Gamma,r}(\Omega_1\Omega_2).\,
  \exists \Omega_1'\Omega_2' \in \NF_{\Gamma,r}(\Omega_1) \NF_{\Gamma,r}(\Omega_2)\
  s.t.\ \Omega_1'\Omega_2' \reduces{\Gamma}{r} \Omega$.

\end{lemma}
\begin{proof}

  By constructing such $\Omega_1',\Omega_2'$ from the given derivations.
\end{proof}

We now state and prove the correctness theorems.

\begin{theorem}[Soundness]
\label{thm:sound-implicit}  
If $\Gamma \vdash C_r \Dashv \Xi$ then $\forall \Omega \in \Xi$. $\Gamma \vdash  C_r \dashv \Omega$
\end{theorem}
\begin{proof}
   By induction on the derivation in normal form.
\begin{description}
    \item[Case: $(\fuse I)$] 
      \begin{smallrules}
            \infer[\fuse I]{\Gamma 
      \vdash A_m \fuse B_m \Dashv \NF(\Xi_L\Xi_R)}{\Gamma 
      \vdash A_m \Dashv {\Xi_L} & \Gamma 
      \vdash B_m \Dashv {\Xi_R}}
    \end{smallrules}
    \begin{tabbing}
        $\forall \Omega_1 \in \Xi_1. \Gamma \vdash A_m \dashv \Omega_1$ \` by $IH(\mathcal{D})$\\
        $\forall \Omega_2 \in \Xi_2. \Gamma \vdash A_m \dashv \Omega_2$ \` by $IH(\mathcal{E})$\\
        $\forall \Omega_1\Omega_2 \in \Xi_1\Xi_2. \Gamma \vdash A_m \fuse B_m \dashv \Omega_1\Omega_2$ \` by $\fuse I$ \\
        $\Xi_1\Xi_2 \reduces{\Gamma}{m} \Xi$ \` by \autoref{thm:complete-restriction} \\ 
        $\forall \Omega \in \Xi. \Gamma \vdash A_m \fuse B_m \dashv \Omega$ \` by mimicking context reduction in rules
    \end{tabbing}
\end{description}
\end{proof}

\begin{theorem}[Completeness]
  \label{thm:complete-implicit}\mbox{}\newline
  If $\Gamma \vdash C_r \dashv \Omega$ then $\Gamma \vdash C_r \Dashv \Xi$ for a
  $\Xi$ with $\NF_{\Gamma,r}(\Omega) \subseteq \Xi$
\end{theorem}
\begin{proof}
  By induction on the derivation, with some applications of \autoref{lemma:join-ok}.
    \begin{description}
    \item[Case:]  $\fuse I$
      \begin{rules}
        \infer[\fuse I]{\Gamma \vdash  A_m \fuse B_m \dashv \Omega_1\Omega_2}{\deduce{\Gamma \vdash A_m \dashv \Omega_1}{\mathcal{D}} & \deduce{\Gamma \vdash B_m \dashv \Omega_2}{\mathcal{E}}}
      \end{rules}
      \begin{tabbing}
        $\Gamma \vdash A_m \Dashv \Xi_1$ and $\NF(\Omega_1) \subseteq \Xi_1$ \` by $IH(\mathcal{D})$\\
        $\Gamma \vdash B_m \Dashv \Xi_2$ and $\NF(\Omega_2) \subseteq \Xi_2$ \` by $IH(\mathcal{E})$\\
        $\Gamma \vdash A_m \fuse B_m \Dashv \Xi$ where $\Xi = \NF(\Xi_1\, \Xi_2)$ \` by (implicit) $\fuse I$\\
        $\forall \Omega \in \NF(\Omega_1\Omega_2).\, \exists \Omega_1'\Omega_2' \in \NF(\Omega_1)\NF(\Omega_2)\ s.t.\ \Omega_1'\Omega_2' \reduces{\Gamma}{m} \Omega$\` by \autoref{lemma:join-ok}\\
        $\forall \Omega \in \NF(\Omega_1\Omega_2).\, \Omega \in \Xi$ \` from above \\
        $\NF(\Omega_1\Omega_2) \subseteq \Xi$ \` by definition
      \end{tabbing}
    \end{description}
\end{proof}

\section{Extending to a Type System}
\label{sec:types}
It is a mostly mechanical process to extend the natural deduction logic system to a bidirectional type system as per \cite{Dunfield22surveys} for a functional programming language. There is one decision to make however. 
Ordered logic has two kinds of implication and two kinds of pairs. The choice we must make here is whether these types should have separate proof terms as well or whether the set of proof terms should remain the same as in the unordered setting found in \cite{Jang24fscd}.
We make the latter choice. We do so to maximize the potential for mode polymorphism in a programming language with this type system. Left and right implication collapse to the un-ordered implication in the presence of full mobility and fuse and twist collapse to the unordered pair.
By leaving $\lambda x.e$ as the proof term for both kinds of implication and $(e_1,e_1)$ as the proof term for both kinds of pairs, programs themselves remain mode agnostic.
For these reasons, we keep the term language from \cite{Jang24fscd}.
We walk through a few of the rules to show the mechanics of constructing a bidirectional system for an ordered type system. We leave the full exposition of such a language for future work as a programming language based on this logic will become more practical once an efficient algorithm is developed.
In a bidirectional type system, there are two judgements: checking written as $e \chk A_m$, in which we view the type as an input, and synthesis $s \syn A_m$ in which we view the type as an output. Introduction rules check the type of the principal type constructor, while elimination rules synthesize it.

We begin with the rules for fuse:
\begin{smallrules}
  \infer[\fuse I]{\Gamma \vdash (e_1,e_2) \chk A_m \fuse B_m \Dashv \NF(\Xi_L\Xi_R)}{\Gamma \vdash e_1 \chk A_m \Dashv \Xi_L & \Gamma \vdash e_2 \chk B_m \Dashv \Xi_R}\\[1em]
  \infer[\fuse E]
  {\Gamma \vdash \textbf{match}\ s\ \textbf{with} (x,y) \Rightarrow e \chk C_r \Dashv \NF(\Xi_L\Xi_M\Xi_R)}
  {\Gamma \vdash s \syn A_m \fuse B_m \Dashv \Xi_M & m \geq r
    & \begin{array}[b]{c}
        \Gamma,x:A_m,y:B_m \vdash e \chk C_r \Dashv \Xi \\
        \Xi{\bbar}\Xi_L(x:A_m)(y:B_m)\Xi_R
      \end{array}}
\end{smallrules}
Then for twist:
\begin{smallrules}
  \infer[\twist I]{\Gamma \vdash (e_1,e_2) \chk A_m \twist B_m \Dashv \NF(\Xi_L\Xi_R)}{\Gamma \vdash e_1 \chk A_m \Dashv \Xi_R & \Gamma \vdash e_2 \chk B_m \Dashv \Xi_L}\\[1em]
  \infer[\twist E]
  {\Gamma \vdash \textbf{match}\ s\ \textbf{with} (x,y) \Rightarrow e \chk C_r \Dashv \NF(\Xi_L\Xi_M\Xi_R)}
  {\Gamma \vdash s \syn A_m \twist B_m \Dashv \Xi_M & m \geq r
    & \begin{array}[b]{c}
        \Gamma,x:A_m,y:B_m \vdash e \chk C_r \Dashv \Xi \\
        \Xi{\bbar}\Xi_L(y:B_m)(x:A_m)\Xi_R
      \end{array}}
\end{smallrules}
Then for right implication: 
\begin{smallrules}
  \infer[\rightimp I]{\Gamma \vdash \lambda x.e \chk A_m \rightimp B_m \Dashv \Xi}{\Gamma,x:A_m \vdash e \chk B_m \Dashv \Xi' & \Xi'{\bbar}\Xi(x:A_m)}\\[1em]
  \infer[\rightimp E]{\Gamma \vdash s\ e \syn B_m \Dashv \NF(\Xi_L\Xi_R)}{\Gamma \vdash s \syn A_m \rightimp B_m \Dashv \Xi_L & \Gamma \vdash e \chk A_m \Dashv \Xi_R}
\end{smallrules}
And finally for left implication:
\begin{smallrules}
  \infer[\leftimp I]{\Gamma \vdash \lambda x.e \chk A_m \rightimp B_m \Dashv \Xi}{\Gamma,x:A_m \vdash e \chk B_m \Dashv \Xi' & \Xi'{\bbar}(x:A_m)\Xi}\\[1em]
  \infer[\leftimp E]{\Gamma \vdash s\ e \syn B_m \Dashv \NF(\Xi_L\Xi_R)}{\Gamma \vdash s \syn A_m \leftimp B_m \Dashv \Xi_R & \Gamma \vdash e \chk A_m \Dashv \Xi_L}
\end{smallrules}

As a small example, consider the expression
\[\lambda x.\lambda y. (x,y)\]
This has four types which would type check if the mode $m$ does not admit any sort of mobility:
\[
  \begin{array}{l}
    A_m \rightimp B_m \rightimp A_m \fuse B_m \\
    A_m \leftimp B_m \leftimp A_m \twist B_m \\
    A_m \rightimp B_m \leftimp A_m \twist B_m \\
    A_m \leftimp B_m \rightimp A_m \fuse B_m
  \end{array}
\]
And more if any kind of mobility is admissible. However, because the proof term remains the same, 
we can just indicate which kind of implication we mean by specifying a type at a top level, and not worry 
about also adjusting the proof term if we want the program to be in an unordered setting.

\section{Discussion and Further Related Work}
\label{sec:conclusion}

The closest related work is that of Kanovich et al.\
\cite{Kanovich19mscs}, who develop a mixed-mode ordered
logic using subexponentials. We build on this by providing a sequent calculus
with modes uniformly governing structural properties, as in
\cite{Pruiksma18un,Jang24fscd}, proving cut elimination for this system, and
establishing decidability for a natural deduction system with implicit
structural rules that corresponds to the initial sequent calculus.  The
subexponential system of \cite{Kanovich18ijcar} embeds into our framework via
$!^m A_{\mmode{O}} \equiv \down^m_{\mmode{O}} \up_{\mmode{O}}^m A_{\mmode{O}}$
where ${\mmode{O}}$ is a base mode that admits no structural properties and
$m \geq {\mmode{O}}$.
 
Directional mobility has been used to model security properties
\cite{Gouni26popl} in a semantic rendering of ordered adjoint types in the
QRTT language. Our algorithm for computing normal forms can be used to decide
\emph{subsumption}, which leads to decidability for a bidirectional
type-checker for QRTT\@.  We cannot fully develop this here, just provide a
small example.  We want to ensure that a high security operation can only be
performed after authorization.  For this purpose, we introduce three modes:
$\ms{low}$, $\ms{auth}$, and $\ms{high}$.  A variable $x{:}1_{\ms{low}}$
represents a low security operation, a variable $x{:}1_{\ms{auth}}$ represents a
(successful) authorization, and $x{:}1_{\ms{high}}$ represents a high security
operation.  These modes have the following structural properties: $\ms{auth}$ is
left mobile (authorization can always occur earlier), $\ms{high}$ is right
mobile (we can always perform a high security operation later), and $\ms{low}$
is mobile in both directions (we can always perform a low security operation).
Then sequences such as
$(x_1{:}1_{\ms{low}})\, (x_2{:}1_{\ms{auth}})\, (x_3{:}1_{\ms{high}})
\,(x_4{:}1_{\ms{low}})$ are allowed, but we cannot move $x_2$ to the right of
$x_3$, or $x_3$ to the left of $x_2$.

Future work includes developing and evaluating an efficient type-checking algorithm. Although
\autoref{sec:nd-implicit} yields a decision procedure, this approach may
be of exponential complexity due to unnecessary weakenings and variable
placements. Possible optimizations include delaying weakening until a variable
exits scope or multiple contexts need to match (e.g., in $\with I$). We could
also track fully mobile propositions separately, as they could occur
anywhere. 

\medskip
\noindent\textbf{Acknowledgments.} This research was sponsored through generous donations from the Jane Street Corporation and a Litton Fellowship in Computer Science. We would also like to thank Klaas Pruiksma for
discussions related to ordered adjoint logic.

\bibliography{fp-copy}

\end{document}